%% file: CoordinatedBeamSelection.tex
\documentclass[12pt, draftclsnofoot, onecolumn, romanappendices]{IEEEtran}

\usepackage{amsmath}
\usepackage{amsthm}
\usepackage{amsbsy}
\usepackage{amssymb}
\usepackage{bm}
\usepackage[noadjust]{cite}
\usepackage{graphicx}
\usepackage{float}
\usepackage{epsfig}
\usepackage{microtype}
\usepackage{mathtools}
\usepackage{latexsym}
\usepackage{wasysym}
\usepackage[subtle]{savetrees}
\usepackage{placeins}
\usepackage{booktabs}
\usepackage{color}
\usepackage{rotating}
\usepackage[usenames,dvipsnames]{xcolor}
\usepackage{caption}
\usepackage{subcaption}
\usepackage{balance}
\usepackage{epstopdf}
\usepackage{cancel}
\usepackage{tikz}
\usetikzlibrary{fit,positioning,arrows.meta}
\usepackage{tkz-euclide}
\usetkzobj{all}
\usepackage{algorithm}
\usepackage{algpseudocode}
\usepackage{overpic}
\usepackage{dashrule}
\usepackage{subdepth}
\usepackage{overpic}
\usepackage{stackengine}
\usepackage{accents}
\usepackage{stmaryrd}
\usepackage{xpatch}

\input{Definitions}
\newcommand{\ubar}[1]{\text{\b{$#1$}}}
\IEEEoverridecommandlockouts
\definecolor{orange}{rgb}{0.8627,0.4314,0.1961}
\definecolor{grey}{rgb}{0.2196,0.2471,0.3176}
\definecolor{blue}{rgb}{0.0275,0.4431,0.5294}
\definecolor{green}{rgb}{0,0.6078,0.4392}
\newcommand{\Mod}[1]{\ (\mathrm{mod}\ #1)}

\begin{document}

\bstctlcite{IEEEexample:BSTcontrol}
\title{Coordinated Beam Selection in Millimeter Wave MU-MIMO Using Out-of-Band Information}
\author{
     	\IEEEauthorblockN{
     		Flavio Maschietti, 
     		David Gesbert,
     		Paul de Kerret
	}
	\IEEEauthorblockA{
     		Communication Systems Department, EURECOM, Sophia-Antipolis, France\\
    		Email: \{flavio.maschietti, david.gesbert, paul.dekerret\}@eurecom.fr}
    	} 
\maketitle

\begin{abstract} 
	Using out-of-band (OOB) side-information has recently been shown to accelerate beam selection in single-user 
	millimeter wave (mmWave) massive MIMO (m-MIMO) communications. In this paper, we propose a novel OOB-aided 
	beam selection framework for a mmWave uplink multi-user system. In particular, we exploit spatial information extracted 
	from lower (sub-$6$ GHz) bands in order to assist with an inter-user coordination scheme at mmWave bands.
	To enforce coordination, we propose an exchange protocol exploiting device-to-device (D2D) communications.
	In particular, low-rate beam-related information is exchanged between the mobile terminals.
	The decentralized coordination mechanism allows the suppression of the so-called co-beam interference which would 
	otherwise lead to irreducible interference at the base station (BS) side, thereby triggering substantial spectral efficiency (SE) gains. 
\end{abstract}
		
\section{Introduction}
	
	The large bandwidths available at mmWave carrier frequencies are expected to help meet the throughput requirements for future 
	mobile networks~\cite{Heath2016}. In order to guarantee appropriate link margins and coverage in response to stronger path 
	losses~\cite{Akdeniz2014}, m-MIMO antennas are expected to be used at both BS and UE sides (when the form factor allows).
	However, configuring those antennas entails an additional effort. The high cost and power consumption of the radio components impact
	on the UEs and small BSs, thus limiting the practical implementation of a \emph{fully-}digital beamforming architecture~\cite{Heath2016}. 
	As a consequence, mixed analog-digital \emph{(hybrid)} architectures have been proposed~\cite{Alkhateeb2014}, where a low-dimensional 
	digital processor is concatenated with an RF analog beamformer, implemented through phase shifters~\cite{Heath2016}.
	
	Interestingly, most works on such architectures opt to leave aside multi-user interference issues in the analog domain and cope with them in 
	the digital part instead. For instance, in \cite{Alkhateeb2015}, the analog stage is intended to find the best beam directions at each UE regardless 
	of the fact that resulting paths arriving at the BS from different UE might end up in the same receive BS analog beam (so-called \emph{co-beam}
	interference). The strength of this approach lies in the fact that it is possible to use the existing beam training algorithms for single-user links -- 
	such as~\cite{Maschietti2017, Garcia2016, Ali2018} -- in the analog stage. Such algorithms have been developed bearing in mind the need for 
	fast link establishment in mmWave communications~\cite{Giordani2019}. Yet, multiple \emph{closely-located} UEs bear certain risk to share one 
	or more common reflectors, causing the potential alignment of some strong paths' angles of arrival (AoA) at the BS~\cite{Akdeniz2014}. 
	In this case, the application of the Zero-Forcing (ZF) criterion on the resulting effective channel in the digital domain might not be effective.
				
	To solve the irreducible uplink co-beam interference problem, a possible approach consists in addressing the interference before it takes place, 
	i.e. the UE side, as is done e.g. in \cite{Zhu2017}. Although showing significant gains over the existing solutions, such works assume perfect CSI 
	for analog beamforming, which might not be realistic in some mmWave contexts~\cite{Alkhateeb2015}.
	
	To go around this problem, we propose a UE \emph{coordination mechanism} exploiting statistical OOB information. Several prior works 
	have pioneered the idea of exploiting side-information (in particular, extracted from sub-$6$ GHz bands) for mmWave performance 
	optimization~\cite{Maschietti2017, Garcia2016, Ali2018}, but -- to our best knowledge -- not in the multi-user setting. 
	The coordination mechanism is based on the idea of each UE \emph{autonomously} selecting an analog beam for transmission so as to 
	strike a trade-off between (i) capturing enough channel gain and (ii) ensuring the UE signals impinge on distinct beams at the BS side. 
	The intuition behind point (ii) is to ensure that the effective channel matrix seen by the BS preserves full rank properties, 
	thus enabling inter-UE interference mitigation in the digital domain. 
	
	In this paper, further novelty originates from (i) the way OOB-based side-information is exploited to enable a coordination mechanism
	between the UEs, and the fact that (ii) not all the UEs need to be endowed with the same amount of side-information. In particular, 
	our scheme leverages a hierarchical information exchange which allows halving of the overall information overhead compared with a full 
	exchange scenario. In this setup, some higher ranked UEs receive sub-$6$ GHz beam information from lower ranked ones only.
	This configuration can be obtained through e.g. D2D communications. In this respect, the 3GPP Release 16 is expected to support point-to-point
	side-links which facilitate cooperative communications among the UEs with low resource consumption~\cite{3GPP2018d}.
		
\section{System Model and Problem Formulation}

We consider a multi-band scenario, where a conventional wireless network using sub-$6$ GHz bands coexists with a mmWave
one. In the following, we introduce the mmWave model. In line with~\cite{Ali2018}, the sub-$6$ GHz model is likewise defined,
with all variables underlined to distinguish them.
				
	\subsection{Uplink Millimeter Wave Model}
			
		The BS is equipped with $N_{\text{BS}} \gg 1$ antennas to support $K$ UEs with $N_{\text{UE}} \gg 1$ antennas each.
		The UEs are assumed to reside in a disk of radius $r$, which will be used to control inter-UE average distance.
		To ease the notation, we assume that the BS has $K$ RF chains available (one for each UE), connected to all the 
		$N_{\text{BS}}$ antennas (fully-connected\footnote{Although \emph{partially-connected} architectures are more 
		relevant for practical implementation due to less stringent hardware requirements~\cite{Park2017}, we assume 
		\emph{fully-connected} architectures to keep notation light, as in most prior works focusing on mmWave SE maximization, 
		e.g.~\cite{Alkhateeb2014, Alkhateeb2015, Zhu2017}. The beam selection strategies which we will propose in Section 
		\ref{sec:OOB_BS} are in principle extendible to all mixed analog/digital beamforming architectures.}hybrid 
		architecture~\cite{Heath2016}).
		
		The $u$-th UE precodes the data $x^{u} \sim \mathcal{CN}(0, 1)$ with the analog unit norm vector
		$\mathbf{v}^{u} \in \mathbb{C}^{N_{\text{UE}} \times 1}$. We assume that the UEs have one RF chain each, 
		i.e. UEs are limited to analog beamforming via phase shifters (constant-magnitude elements)~\cite{Alkhateeb2014}.
		In addition, $\mathbb{E}[\norm{\mathbf{v}^u x^u}^2] \le 1$, assuming normalized power constraints.
		The reconstructed signal after mixed analog/digital combining at the BS is expressed as follows 
		-- assuming no timing and carrier mismatches:
		\begin{equation} \label{RX_Signal_pAnalogComb}
			\mathbf{\hat{x}} = \mathbf{W}_{\text{D}} \sum_{u=1}^K \mathbf{W}_{\text{RF}}^{\textrm{H}} 
			\mathbf{H}^{u} \mathbf{v}^{u} x^{u} + \mathbf{W}_{\text{D}} \mathbf{W}_{\text{RF}}^{\textrm{H}} \mathbf{n}
		\end{equation}
		where $\mathbf{H}^{u} \in \mathbb{C}^{N_{\text{BS}} \times N_{\text{UE}}}$ is the channel matrix from the $u$-th UE 
		to the BS, $\mathbf{n} \sim \mathcal{CN}(\mathbf{0}, \sigma_\mathbf{n}^2 \mathbf{I})$ is the thermal noise vector,
		$\mathbf{W}_{\text{RF}} \in \mathbb{C}^{N_{\text{BS}} \times K}$ contains the beamformers relative to
		each RF chain (subject to the same hardware constraints as described above),
		and $\mathbf{W}_{\text{D}} \in \mathbb{C}^{K \times K}$ denotes the digital combining matrix.
		
		Introducing the effective channel $\mathbf{h}_{\text{e}}^u = \mathbf{W}_{\text{RF}}^{\textrm{H}} 
		\mathbf{H}^u \mathbf{v}^u \in \mathbb{C}^{K \times 1}$ of the $u$-th UE, we can write
		\eqref{RX_Signal_pAnalogComb} as follows:
		\begin{equation}
			\mathbf{\hat{x}} = 
			\mathbf{W}_{\text{D}} \sum_{u=1}^K \mathbf{h}_{\text{e}}^u x^u + \mathbf{W}_{\text{D}} \mathbf{\tilde{n}} =
			\mathbf{W}_{\text{D}} \mathbf{H}_{\text{e}} \mathbf{x} + \mathbf{W}_{\text{D}} \mathbf{\tilde{n}}
		\end{equation}
		where $\mathbf{H}_{\text{e}} \in \mathbb{C}^{K \times K}$ denotes the effective channel matrix -- 
		containing all the single effective channels -- and where $\mathbf{\tilde{n}} = 
		\mathbf{W}_{\text{RF}}^{\textrm{H}} \mathbf{n}$ denotes the filtered thermal noise vector. \vspace{-0.21cm}
		
		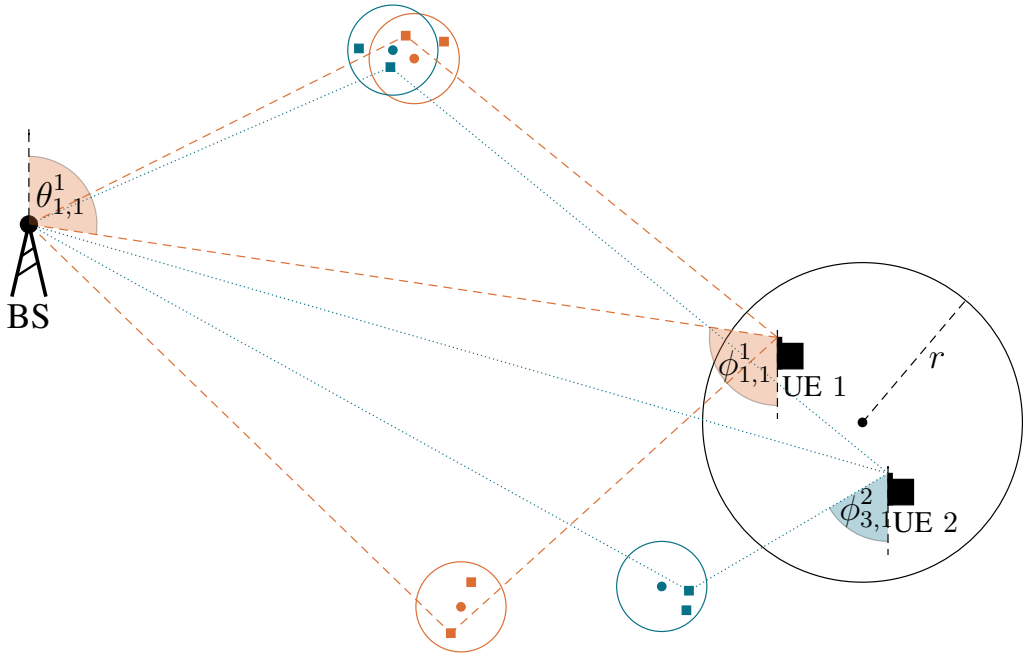
\begin{figure}[h]
			\centering
			\resizebox{14cm}{!}{
			\begin{tikzpicture}
				\draw (5,-1.18) circle (1.88cm);
				\filldraw[black] (5,-1.18) circle (.05cm);
				\draw[densely dashed] (5,-1.18) -- (6.25, .3);
				\draw (5.88,-.45) node {$r$};
				\filldraw[black] (4,-.55) rectangle +(.3,.3);
				\filldraw[black] (4,-.55) rectangle +(.05,.37);
				\draw (4.44,-.76) node {\small{UE $1$}};
				\filldraw[black] (5.3,-2.15) rectangle +(.3,.3);
				\filldraw[black] (5.3,-2.15) rectangle +(.05,.37);
				\draw (5.74,-2.36) node {\small{UE $2$}};
				\draw[very thick] (-5,0.3) -- (-4.8,1.15);
				\draw[very thick] (-4.8,1.15) -- (-4.6,.3);
				\draw[thick] (-4.88,.78) -- (-4.73,.88);
				\draw[thick] (-4.93,.55) -- (-4.7,.7);
				\draw (-4.8, .07) node {BS};
				\draw[orange] (-.27,3.1) circle (.53cm);
				\filldraw[orange] (-.27, 3.1) circle (.05cm);
				\filldraw[orange] (.03, 3.25) rectangle +(.1, .1);
				\filldraw[orange] (-.42, 3.32) rectangle +(.1, .1);
				\draw[blue] (-.521,3.2) circle (.53cm);
				\filldraw[blue] (-.521, 3.2) circle (.05cm);
				\filldraw[blue] (-.97, 3.17) rectangle +(.1, .1);
				\filldraw[blue] (-.6, 2.95) rectangle +(.1, .1);
				\draw[orange] (.28,-3.35) circle (.53cm);
				\filldraw[orange] (.28, -3.35) circle (.05cm);
				\filldraw[orange] (.11, -3.71) rectangle +(.1, .1);
				\filldraw[orange] (.35, -3.11) rectangle +(.1, .1);
				\draw[blue] (2.64,-3.11) circle (.53cm);
				\filldraw[blue] (2.64, -3.11) circle (.05cm);
				\filldraw[blue] (2.91, -3.21) rectangle +(.1, .1);
				\filldraw[blue] (2.88, -3.44) rectangle +(.1, .1);
				\draw[orange, densely dashed] (4,-.18) -- (-.37,3.37);
				\draw[orange, densely dashed] (-4.8, 1.15) -- (-.37,3.37);
				\draw[orange, densely dashed] (4,-.18) -- (.16,-3.66);
				\draw[orange, densely dashed] (-4.8, 1.15) -- (.16,-3.66);
				\draw[orange, densely dashed] (4,-.18) -- (-4.8,1.15);
				\draw[blue, densely dotted] (5.3,-1.78) -- (-.55,3);
				\draw[blue, densely dotted] (-4.8, 1.15) -- (-.55,3);
				\draw[blue, densely dotted] (5.3,-1.78) -- (2.96,-3.16);
				\draw[blue, densely dotted] (-4.8, 1.15) -- (2.96,-3.16);
				\draw[blue, densely dotted] (5.3,-1.78) -- (-4.8,1.15);
				\filldraw[black] (-4.8,1.15) circle (.1cm); 
				\coordinate (C) at (4,-.18);
				\coordinate (B) at (-4.8,1.15);
				\coordinate (A) at (-4.8,3.35);
				\tkzMarkAngle[fill=orange,size=0.8cm,opacity=.3](C,B,A)
				\tkzLabelAngle[pos = 0.5](A,B,C){$\theta_{1,1}^1$}
				\coordinate (C) at (4,-.18);
				\coordinate (B) at (-4.8,1.15);
				\coordinate (A) at (4,-3.35);
				\tkzMarkAngle[fill=orange,size=0.8cm,opacity=.3](B,C,A)
				\tkzLabelAngle[pos = -0.5](B,C,A){$\phi_{1,1}^1$}
				\coordinate (C) at (5.3,-1.78);
				\coordinate (B) at (2.96,-3.16);
				\coordinate (A) at (5.3,-3.35);
				\tkzMarkAngle[fill=blue,size=0.8cm,opacity=.3](B,C,A)
				\tkzLabelAngle[pos = 0.5](B,C,A){$\phi_{3,1}^2$}
				\draw[densely dashed](-4.8, 1.24) to[out=90,in=90] (-4.8,2.2);
				\draw[densely dashed](4, -.18) to[out=90,in=90] (4, -1.14);
				\draw[densely dashed](5.3, -1.78) to[out=90,in=90] (5.3, -2.74);
			\end{tikzpicture}
			}
			\caption{Co-beam interference example with $C = 3$ clusters, $L = 2$ paths, and $K = 2$ UEs.
			The UEs are assumed to reside in a disk of radius $r$. In this illustration, two \emph{closely-located} UEs share some 
			reflectors and the signal waves reflecting on the top ones arrive quasi-aligned at the BS -- i.e. captured with the same BS beam -- 
			while originating from distinct UEs.}
			\label{fig:Scen}
		\end{figure}

	\subsection{Channel Model}

		Assuming mmWave channels exhibit limited scattering~\cite{Akdeniz2014}, we adopt a geometric narrowband channel
		model with $C$ clusters, each one contributing to $L$ paths. 
		The channel matrix $\mathbf{H}^{u} \in \mathbb{C}^{N_{\textrm{BS}} \times N_{\textrm{UE}}}$ for the $u$-th UE is 
		thus expressed as follows~\cite{Ali2018}:
		\begin{equation} \label{H}
			\mathbf{H}^{u} \triangleq \sqrt{N_{\textrm{BS}} N_{\textrm{UE}}} 
			\left(\sum_{c=1\vphantom{\ell}}^C \sum_{\ell=1}^L \alpha_{c, \ell}^{u}
			\mathbf{a}_{\textrm{BS}}(\theta_{c,\ell}^{u})\mathbf{a}^{\textrm{H}}_{\textrm{UE}} (\phi_{c,\ell}^{u}) \right)
		\end{equation}
		where $\alpha_{c, \ell}^{u} \sim \mathcal{CN}(0, \sigma_c^2)$ denotes the complex gain for the $\ell$-th path of the 
		$c$-th cluster of the $u$-th UE, including the shaping filter and the large-scale pathloss. 
		The variables $\phi_{c,\ell}^{u} \in [0, 2\pi)$ and $\theta_{c,\ell}^{u} \in [0, 2\pi)$ are the AoD and AoA for the $\ell$-th
		path of the $c$-th cluster connecting the $u$-th UE to the BS. 
		The vectors $\mathbf{a}_{\textrm{UE}}(\cdot) \in \mathbb{C}^{N_{\textrm{UE}} \times 1}$
		and $\mathbf{a}_{\textrm{BS}} (\cdot) \in \mathbb{C}^{N_{\textrm{BS}} \times 1}$ 
		denote the antenna \emph{unitary} steering vectors at the $u$-th UE and the BS, respectively.
		We assume \emph{uniform linear arrays} (ULA) with $\lambda/2$ inter-element spacing.
		
\subsection{Analog Codebooks}

	We define the codebooks used for analog beamforming as
	\begin{equation}
		\mathcal{V} \triangleq \{ \mathbf{v}_1, \dots, \mathbf{v}_{M_{\text{UE}}} \}, \qquad
		\mathcal{W} \triangleq \{ \mathbf{w}_1, \dots, \mathbf{w}_{M_{\text{BS}}} \}
	\end{equation}
	where $M_{\text{UE}} = N_{\text{UE}}$ and $M_{\text{BS}} = N_{\text{BS}}$ denote the number of elements (beamforming
	vectors) in the codebooks, and where $\mathcal{V}$ is assumed to be shared between all the UEs, to ease the notation.
	
	For instance, with ULA, a suitable design for the fixed elements in the codebook consists in 
	selecting steering vectors over a discrete grid of angles, as follows~\cite{Alkhateeb2015}:
	\begin{equation} \label{AnBeamform1}
		\mathbf{v}_n = \mathbf{a}_{\textrm{UE}}(\hat{\phi}_n), \quad n \in \llbracket 1, M_{\text{UE}} \rrbracket
	\end{equation}
	\begin{equation}  \label{AnBeamform2}
		\mathbf{w}_m = \mathbf{a}_{\text{BS}}(\hat{\theta}_m), \quad m \in \llbracket 1, M_{\text{BS}}\rrbracket
	\end{equation}
	where the quantized angles $\hat{\phi}_n$ and $\hat{\theta}_m$ can be chosen according to different sampling strategies 
	of the $[0, \pi]$ range~\cite{Maschietti2017}.
	\begin{remark}
		The notation $\llbracket 1, M \rrbracket$ denotes the set $\{ 1, \dots, M \}$. 
		The same notation will be used in the remainder of the paper. \qed
	\end{remark}

\subsection{Problem Formulation} \label{sec:OOB_BS}

	The beam selection problem in mmWave communications consists in selecting the analog transmit and receive beams from 
	$\mathcal{V}$ and $\mathcal{W}$ to maximize the sum-rate defined as follows:
	\begin{equation} \label{Sum_Rate}
		R(\mathbf{n}, \mathbf{m}) \triangleq \sum_{u=1}^K \log_2 \big( 1+\gamma^u(\mathbf{n}, \mathbf{m}) \big)
	\end{equation}
	where $\mathbf{n} \triangleq \big[n_1 \quad \dots \quad n_K\big]$ 
	(resp. $\mathbf{m} \triangleq \big[m_1 \quad \dots \quad m_K	\big]$) is the vector containing the selected
	beams at the UE  (resp. BS side), while $\gamma^u$ is the received SINR for the $u$-th UE, defined as~\cite{Alkhateeb2015}
	\begin{equation} \label{SINR}
		\gamma^{u}(\mathbf{n}, \mathbf{m}) \triangleq		\frac{|\mathbf{w}_{\text{D}}^u \mathbf{h}_{\text{e}}^u|^2}
		{\sum_{w \ne u}|\mathbf{w}_{\text{D}}^u \mathbf{h}_{\text{e}}^{w}|^2 + 
		\norm{\mathbf{w}_{\text{D}}^u}^2 \sigma_{\mathbf{\tilde{n}}}^2}
	\end{equation}
	with $\mathbf{w}_{\text{D}}^{u} \in \mathbb{C}^{1 \times K}$ denoting the $u$-th row of $\mathbf{W}_{\text{D}}$.
		
	In order to maximize \eqref{Sum_Rate}, the mutual optimization of both analog and digital components must be considered. 
	A common viable approach consists in decoupling the design, as the analog precoder can be optimized through long-term
	statistical information, whereas the digital one can be made dependent on instantaneous one~\cite{Alkhateeb2015}. 
	The same approach is followed here. 
	
	In particular, we consider ZF combining, so that we have
	\begin{equation}
		\mathbf{W}_{\text{D}} = 
		\big(\mathbf{H}_{\text{e}}^{\textrm{H}} \mathbf{H}_{\text{e}} \big)^{-1} \mathbf{H}_{\text{e}}^{\textrm{H}}.
	\end{equation}
	The received SINR for the $u$-th UE is then simplified as
	\begin{equation} \label{SINRZF}
		\gamma^u(\mathbf{n}, \mathbf{m}) = 
		\frac{1}{\sigma_{\mathbf{\tilde{n}}}^2 \{\big(\mathbf{H}_{\text{e}}^{\textrm{H}} \mathbf{H}_{\text{e}}\big)^{-1}\}_{u, u}},
	\end{equation}
	with $\{\cdot\}_{u, u}$ denoting the $u$-th element on the diagonal of 
	$(\mathbf{H}_{\text{e}}^{\textrm{H}} \mathbf{H}_{\text{e}})^{-1}$, associated to the $u$-th UE.
	
	In general, the \emph{perfect} knowledge of the effective channels plus a \emph{centralized} operator to instruct the UEs 
	are needed to maximize \eqref{Sum_Rate} via \eqref{SINRZF}. Such information is not available without a significant 
	resource overhead. In the next section, we propose some strategies to exploit sub-$6$ GHz 
	information for a distributed and low-overhead approach to the problem.
	
\section{Out-of-Band-Aided Beam Selection} \label{sec:OOB_BS}
	
	Let us consider the existence of a sub-$6$ GHz channel $\ubar{\mathbf{H}}^u \in \mathbb{C}^{\ubar{N}_{\textrm{BS}}
	\times \ubar{N}_{\textrm{UE}}}$ between the $u$-th UE and the BS. We assume that each UE is able to compute a
	\emph{spatial spectrum} $\mathbb{E}[|\ubar{\mathbf{S}}^u|^2] \in \mathbb{C}^{\ubar{M}_{\textrm{BS}} \times 
	\ubar{M}_{\textrm{UE}}}$ of the sub-$6$ GHz channel, where~\cite{Ali2018}
	\begin{equation}
		\ubar{\mathbf{S}}^u = \ubar{\mathbf{W}}^{\mathrm{H}} \ubar{\mathbf{H}}^u \ubar{\mathbf{V}}
	\end{equation}
	and where the expectation is over fast fading.
	The matrices $\ubar{\mathbf{W}} \in \mathbb{C}^{\ubar{N}_{\textrm{BS}} \times \ubar{M}_{\textrm{BS}}}$ and 
	$\ubar{\mathbf{V}} \in \mathbb{C}^{\ubar{N}_{\textrm{UE}} \times \ubar{M}_{\textrm{UE}}}$ collect all the sub-$6$ GHz
	steering vectors at the BS and UE sides, sampled at the same angles as the mmWave ones. In particular, we assume 
	$\ubar{N}_{\textrm{BS}} \ll \ubar{M}_{\textrm{BS}} = M_{\textrm{BS}}$ and 
	$\ubar{N}_{\textrm{UE}} \ll \ubar{M}_{\textrm{UE}} = M_{\textrm{UE}}$. 
	The $(\ubar{m}, \ubar{n})$-th element of $\mathbb{E}[|\ubar{\mathbf{S}}^u|^2]$ contains thus the sub-$6$ GHz channel 
	gain obtained with the $\ubar{n}$-th beam at the $u$-th UE and the $\ubar{m}$-th one at the BS.
	\begin{remark}
		The computation of $\mathbb{E}[|\ubar{\mathbf{S}}^u|^2]$ is \emph{merely} bound to the knowledge of the average 
		sub-$6$ GHz channel, as $\ubar{\mathbf{W}}$ and $\ubar{\mathbf{V}}$ are predefined fixed matrices. Note that the 
		acquisition of the CSI matrix for conventional sub-$6$ GHz communications is a standard operation~\cite{3GPP2018c}.
		In this respect, sub-$6$ GHz channel measurements can be collected and stored \emph{periodically} -- e.g. within the
		channel coherence time -- to be \emph{readily} available for evaluating $\mathbb{E}[|\ubar{\mathbf{S}}^u|^2]$.
		In other words, obtaining the spatial spectrum $\mathbb{E}[|\ubar{\mathbf{S}}^u|^2]$ requires no additional 
		training overhead~\cite{Ali2018}. \qed
	\end{remark}
			
	\subsection{Exploiting Sub-6 GHz Information}	
	
		The available sub-$6$ GHz spatial information can be exploited to obtain a rough estimate of the angular
		characteristics of the mmWave channel. Indeed, due to the larger beamwidth of sub-$6$ GHz beams, one sub-$6$ GHz 
		beam can be associated to a set of mmWave beams, as defined below.		
		
		\begin{definition}
			For a given sub-$6$ GHz beam pair $(\ubar{n}, \ubar{m})$, we introduce the set $\mathcal{S}(\ubar{n}, \ubar{m}) 
			\triangleq \mathcal{S}_{\textnormal{UE}}(\ubar{n}) \times \mathcal{S}_{\textnormal{BS}}(\ubar{m})$ where 
			$\mathcal{S}_{\textnormal{UE}}(\ubar{n})$  (resp. $\mathcal{S}_{\textnormal{BS}}(\ubar{m})$) contains all the 
			mmWave beams belonging to the $3$-dB beamwidth of the $\ubar{n}$-th (resp. $\ubar{m}$-th) sub-$6$ GHz beam.
		\end{definition}
		
		It is important to remark that we focus in this work on the selection of sub-$6$ GHz beams to further refine. 
		We indeed adhere to the well-known two-stage beamforming and training operation, where \emph{fine-grained} training 
		(called beam refinement) follows \emph{coarse-grained} training (called sector sweeping). 
		In our approach, coarse-grained beam selection is achieved without \emph{actually} training the beams with reference signals, 
		but using instead beam information extracted from lower channels, so as to speed up the process. 
		Once these coarse sub-$6$ GHz beams are chosen, the small subset of associated mmWave beams is trained. 
		We refer to~\cite{Kim2014} for more details on this standard step. In what follows, we propose some multi-user beam selection 
		strategies leveraging the described OOB-related side-information.
			
	\subsection{Uncoordinated Beam Selection}

		We first describe here an approach based on~\cite{Alkhateeb2015}, where the authors proposed to design the analog
		beamformers so as to maximize the received power (SNR) for each UE, neglecting multi-user interference. When OOB 
		information is available, the beam selection $(\ubar{n}_u^{\text{un}} \in \ubar{\mathcal{V}},\ubar{m}_u^{\text{un}} \in
		\ubar{\mathcal{W}})$ at the $u$-th UE -- which we will denote as \emph{uncoordinated} (un) -- can be expressed as follows:
		\begin{align} \label{UN_Alg}
					(\ubar{n}_{u}^{\text{un}},& \ubar{m}_{u}^{\text{un}}) =
					\argmax_{\ubar{n}_u, \ubar{m}_u} ~\log_2\big(1 + \mathbb{E}_{n_u, m_u | \ubar{n}_u, \ubar{m}_u} 
					\big[ \gamma^u_{\text{su}}(n_u, m_u) \big] \big)
		\end{align}
		where we have approximated the rate via Jensen's inequality and we have defined the single-user expected SNR, conditioned
		on a given sub-$6$ GHz beam pair $(\ubar{n}_u, \ubar{m}_u) \in \ubar{\mathcal{V}} \times \ubar{\mathcal{W}}$, as follows:
		\begin{equation} \label{EqSINR_Un_K}
			\mathbb{E}_{n_u, m_u|\ubar{n}_u, \ubar{m}_u} \big[ \gamma^u_{\textnormal{su}}(n_u, m_u) \big] =
			\hspace{-0.5cm} \sum_{(n_u, m_u) \in \mathcal{S}(\ubar{n}_u, \ubar{m}_u)}
			\hspace{-0.1cm}\frac{g_{n_u, m_u}}{{S_u \sigma_{\mathbf{\tilde{n}}}^2}}
		\end{equation}
		with
		\begin{align}
		g_{n_u, m_u} &\triangleq \mathbb{E} \Big[ \big|\mathbf{w}^u_{m_u} \mathbf{H}^u \mathbf{v}^u_{n_u} \big|^2 \Big] \\
		&= \mathbb{E} \Big[ \big|\mathbf{S}^u_{n_u, m_u}\big|^2 \Big]
		\end{align} 
		being the average beamforming gain obtained at the $u$-th UE with the transmit-receive beam pair $(n_u, m_u)$, 
		and where $S_u\triangleq \text{card}(\mathcal{S}(\ubar{n}_u, \ubar{m}_u))$. 
	
		To solve \eqref{EqSINR_Un_K}, the $u$-th UE needs to know the mmWave gain $g_{n_u, m_u}
		~\forall (n_u, m_u) \in \mathcal{S}(\ubar{n}_u, \ubar{m}_u)$. This information is not available but can be replaced for 
		algorithm derivation purposes\footnote{The proposed algorithms are then evaluated in Section \ref{sec:sims} under 
		realistic multi-band channel conditions as proposed in~\cite{Ali2018}, where the described behavior and consequent 
		randomness is taken into account.} with the gain observed in the sub-$6$ GHz channel over the beam pair 
		$(\ubar{n}_u, \ubar{m}_u)$. In other words, we assume
		\begin{equation} \label{Gain_Approx}
			g_{n_u, m_u} \approx \mathbb{E}\left[\big|\ubar{\mathbf{S}}^u_{\ubar{n}_u, \ubar{m}_u}\big|^2 \right]
			\quad \forall (n_u, m_u) \in \mathcal{S}(\ubar{n}_u, \ubar{m}_u).
		\end{equation}
		Note that the average gain information derived from $\ubar{\mathbf{S}}$ will \emph{unlikely} match with
		its mmWave counterpart in absolute terms practice, due to multipath, noise effects and pathloss discrepancies. 
		Still, high correlation has been observed between the temporal and angular characteristics of the LOS path in 
		sub-$6$ GHz and mmWave channels~\cite{Anjinappa2018}. The correlation diminishes as the LOS condition is lost, 
		as small scattering objects participating in the radio propagation emerge at higher frequencies~\cite{Raghavan2018}.
		Nevertheless, it has been shown in~\cite{Ali2016} that, in an outdoor scenario with strong reflectors (buildings), the paths 
		with uncommon AoA at frequencies far apart\footnote{In~\cite{Ali2016}, $5$ carrier frequencies ranging between $900$ MHz and 
		$90$ GHz have been compared.} are less than $10$\% of the overall paths.
		In this respect, \eqref{Gain_Approx} allows to spot a valuable candidate set for mmWave beams in most of the situations. 
		Yet, an important limitation of this approach is that each UE solves its own beam selection problem \emph{independently} 
		of the other UEs, thus ignoring the possible impairments in terms of interference. Therefore, \emph{as the inter-UE average 
		distance decreases}, the performance of this procedure degrades since the UEs have much more chance to share their 
		best propagation paths -- which results in co-beam interference at the BS.
	
	\subsection{Hierarchical Coordinated Beam Selection}
	
		In order to achieve coordination, we propose to use a hierarchical information structure requiring small 
		overhead. In particular, an \emph{(arbitrary)} order among the UEs is established\footnote{The hierarchical information
		exchange is proposed here to facilitate the coordination mechanism at reduced overhead. In this paper, we shall leave 
		aside further analysis on how such a \emph{hierarchy} is defined and maintained.}, for which the $u$-th UE has access to 
		the beam decisions carried out at the (lower-ranked) UEs $1, \dots, u-1$. This configuration is obtainable through e.g. 
		dedicated D2D channels in the lower bands\footnote{D2D communications allows to exchange information among 
		\emph{closely-located} UEs with low resource (power, time, etc.) consumption~\cite{Chen2017}. In particular, the power 
		consumption for exchanging low-rate beam information over D2D side-links could be negligible due to the small relative 
		path loss as compared to communicating to the BS.}. 
		We further assume that such exchanged beam information is \emph{perfectly} decoded at the intended UEs.
		
		Since the UEs exchange beam indexes (in the order of few bits), the communication overhead is kept low.
		Moreover, the so-called \emph{beam coherence time} -- which depends on beam width and UE speed among others -- has been
		reported to be much longer than the channel coherence time~\cite{Va2017}. As a consequence, such overhead is only generated
		at long intervals.
		\begin{remark} 
			Exchanging sub-$6$ GHz beams rather than mmWave ones introduces some \emph{uncertainty}, but allows to 
			save time as no UE has to wait for another one to perform beam training. \qed
		\end{remark}
				
		Assuming that the sub-$6$ GHz beam indices $\ubar{m}_{{1}, \dots, u-1}$ have been received, the
		\emph{coordinated} (co) sub-$6$ GHz beam pair $(\ubar{n}_u^{\text{co}} \in \ubar{\mathcal{V}},\ubar{m}_u^{\text{co}}
		\in \ubar{\mathcal{W}})$ at the $u$-th UE is obtained through
			\begin{align} \label{Coo_Alg}
				(\ubar{n}&_{u}^{\text{co}}, \ubar{m}_{u}^{\text{co}}) =
				\argmax_{\ubar{n}_u, \ubar{m}_u} 
				~\log_2\left(1 + \mathbb{E}_{\mathbf{n}, \mathbf{m}|\ubar{n}_u,\ubar{m}_{1, \dots, u+1}}
				\big[ \gamma^u(\mathbf{n}, \mathbf{m}) \big] \right).
			\end{align}
		Solving \eqref{Coo_Alg} is not trivial, being a subset selection problem for which a Monte-Carlo approach to approximate the 
		expectation with a discrete summation leads to unpractical computational time. Interestingly, for large $N_{\textnormal{BS}}$ 
		and $N_{\textnormal{UE}}$, we are able to derive an approximation for the expectation in \eqref{Coo_Alg} which will be
		useful for algorithm derivation.  We start with showing such intermediate result. \vspace{-0.21cm}
		
		\begin{proposition} \label{PropSINR_Bin_Expected}
			In the limit of large $N_{\textnormal{BS}}$ and $N_{\textnormal{UE}}$, the expected SINR (averaged over small-scale
			fading) of the $u$-th UE obtained after ZF combining at the BS is
			\begin{equation} \label{SINR_Bin_Expected}
			\mathbb{E} \big[ \gamma^u(\mathbf{n},\mathbf{m}) \big] =\begin{cases}
			\dfrac{g_{n_u, m_u}}{\sigma_{\mathbf{\tilde{n}}}^2} & 
			\textnormal{if} ~m_u \neq m_w~\forall w \in \mathcal{K} \backslash \{ u \} \\
			0 & \textnormal{if} ~\exists ~w \in \mathcal{K} \backslash \{ u \} : m_w = m_u
			\end{cases}
			\end{equation}
			where we have defined $\mathcal{K}\triangleq \llbracket 1,K\rrbracket$.
		\end{proposition} \vspace{-0.47cm}
		
		\begin{proof}
		Based on the result in~\cite{Ngo2014}, we assume that the quantized angles $\hat{\phi}_n, n \in \llbracket 1, M_{\text{UE}} \rrbracket$ 
		and $\hat{\theta}_m, m \in \llbracket 1, M_{\text{BS}} \rrbracket$ are spaced according to the inverse cosine function. 
		The following lemma states an interesting consequence (constant inner product) of such a spacing which will be useful later.
		\begin{lemma} \label{LemOrth}
		Let the angles $\hat{\phi}_n$ and $\hat{\theta}_m$ be spaced according to the 
		inverse cosine function, as follows:
		\begin{equation} \label{Inverse_Cosine}
			\begin{aligned}
				\hat{\phi}_n &= \arccos\Big(1-\frac{2(n-1)}{M_{\textnormal{UE}}-1}\Big), 
				\quad n \in \llbracket 1,M_{\text{UE}}\rrbracket \\
				\hat{\theta}_m &= \arccos\Big(1-\frac{2(m-1)}{M_{\textnormal{BS}}-1}\Big), 
				\quad m \in \llbracket 1,M_{\text{BS}}\rrbracket,
			\end{aligned}
		\end{equation}
		then
		\begin{equation}
			\begin{aligned}
				\mathbf{a}^{\textnormal{H}}_{\textnormal{UE}}(\hat{\phi}_n)\mathbf{a}_{\textnormal{UE}}(\hat{\phi}_{\tilde{n}}) 
				&= 1/N_{\textnormal{UE}}\\
		 		\mathbf{a}^{\textnormal{H}}_{\textnormal{BS}}(\hat{\theta}_m)\mathbf{a}_{\textnormal{BS}}(\hat{\theta}_{\tilde{m}}) 
		 		&=1/N_{\textnormal{BS}}
			\end{aligned}
		\end{equation}
		for any $n \neq \tilde{n}$ and $m \neq \tilde{m}$.
		\end{lemma}
		\begin{proof}[Proof of Lemma \ref{LemOrth}]
			In the following, we will consider w.l.o.g. the UE side.
			
			Let $\Delta \triangleq \cos(\hat{\phi}_n) - \cos(\hat{\phi}_{\tilde{n}})$, then we have:
			\begin{align} \label{ortho_proof_1}
				\mathbf{a}^{\textrm{H}}_{\textrm{UE}}(\hat{\phi}_n) \mathbf{a}_{\textrm{UE}}(\hat{\phi}_{\tilde{n}})
				&= \frac{1}{N_{\textrm{UE}}} \sum_{k=0}^{N_{\textrm{UE}}-1} e^{-i \pi k \Delta}\\
				& \stackrel{(a)}{=} \frac{1}{N_{\textrm{UE}}} \frac{1-e^{-i \pi N_{\textrm{UE}} \Delta}}{1-e^{-i \pi \Delta}}
			\end{align}
			where $(a)$ is due to geometric series properties. 
			
			According to the spacing in \eqref{Inverse_Cosine}, we can write $\Delta= \frac{2(\tilde{n} - n)}{N_{\textrm{UE}}-1}$. 
			Inserting this expression in \eqref{ortho_proof_1} gives:
			\begin{align}
				\mathbf{a}^{\textrm{H}}_{\textrm{UE}}(\hat{\phi}_n) \mathbf{a}_{\textrm{UE}}(\hat{\phi}_{\tilde{n}})
				&= \frac{1}{N_{\textrm{UE}}} \frac{e^{i \pi \Delta} - e^{-i \pi 2(\tilde{n}-n)}}{e^{i \pi \Delta} - 1} \\
				&\stackrel{(b)}{=} \frac{1}{N_{\textrm{UE}}}
			\end{align}
			where $(b)$ follows from $2\pi (\tilde{n}-n)=0 \Mod {2\pi}$ for $n\neq \tilde{n}$.
		\end{proof}
		
		According to Lemma \ref{LemOrth}, in the limit of large $N_{\textnormal{BS}}$ and $N_{\textnormal{UE}}$, 
		$\mathbf{a}_{\textrm{UE}}(\hat{\phi}_n) \perp \text{span}(\mathbf{a}_{\textrm{UE}}(\hat{\phi}_{\tilde{n}}) 
		~\forall \tilde{n} \neq n)$. Likewise $\mathbf{a}_{\textrm{BS}}(\hat{\theta}_m) \perp \text{span}(\mathbf{a}_{\textrm{BS}}
		(\hat{\theta}_{\tilde{m}})~\forall \tilde{m} \neq m)$.
		As a consequence, the matrices
		\begin{equation}
			\hat{\mathbf{A}}_{\text{BS}} = \begin{bmatrix}
			\mathbf{a}_{\textrm{BS}}(\hat{\theta}_1) & \dots & \mathbf{a}_{\textrm{BS}}(\hat{\theta}_{M_{\text{BS}}})
			\end{bmatrix},
		\end{equation}
		and
		\begin{equation}
			\hat{\mathbf{A}}_{\text{UE}} = \begin{bmatrix}
			\mathbf{a}_{\textrm{UE}}(\hat{\phi}_1) & \dots & \mathbf{a}_{\textrm{UE}}(\hat{\phi}_{M_{\text{UE}}})
			\end{bmatrix}
		\end{equation}
		are \emph{asymptotically unitary}. To go further, we resort to the channel approximation in~\cite{Sayeed2002}, 
		which consists in approximating the channel given in \eqref{H} using the quantized angles, as follows:
		\begin{equation}
			\mathbf{H}^u \approx \sqrt{N_{\textrm{BS}} N_{\textrm{UE}}} 
			\Big(\sum_{n=1}^{M_{\text{UE}}} \sum_{m=1}^{M_{\text{BS}}} \psi_{n, m}^{u}
			\mathbf{a}_{\textrm{BS}}(\hat{\theta}_{m})\mathbf{a}^{\textrm{H}}_{\textrm{UE}} (\hat{\phi}_{n}) \Big)
		\end{equation}
		where $\psi_{n, m}^u$ is equal to the sum of the gains of the paths whose angles lie in the \emph{virtual spatial bin}
		centered on $(\hat{\phi}_{n}, \hat{\theta}_{m})$.
		
		We rewrite now \eqref{SINRZF} using the Schur complement as follows: 
		\begin{equation} \label{SINRZF_Beams}
			\gamma^u(\mathbf{n}, \mathbf{m}) = 
			\frac{1}{\sigma_{\mathbf{\tilde{n}}}^2} \big[ (\mathbf{h}_{\text{e}}^{u})^{\mathrm{H}} 
			\mathbf{h}_{\text{e}}^{u}- (\mathbf{h}_{\text{e}}^{u})^{\textrm{H}} 
			\mathbf{P}_{\text{e}/u}  \mathbf{h}_{\text{e}}^{u} \big]
		\end{equation}
		where $\mathbf{P}_{\text{e}/u} \triangleq \mathbf{H}_{\text{e}/u} 
		(\mathbf{H}_{\text{e}/u}^{\textrm{H}}\mathbf{H}_{\text{e}/u})^{-1} \mathbf{H}_{\text{e}/u}^{\textrm{H}}$ 
		is the orthogonal projection onto the $\text{span}(\mathbf{H}_{\text{e}/u})$, with $\mathbf{H}_{\text{e}/u}$ 
		being the submatrix obtained via removing the $u$-th column from $\mathbf{H}_{\mathrm{e}}$.
		 
		Since $\hat{\mathbf{A}}_{\text{UE}}$ and $\hat{\mathbf{A}}_{\text{BS}}$ are asymptotically unitary, it holds that
		\begin{equation}
			\mathbf{P}_{\text{e}/u} \mathbf{h}_{\text{e}}^u = \begin{cases}
				\mathbf{0} & \text{if} ~m_u \neq m_w~\forall w \in \mathcal{K} \backslash \{ u \} \\
				\mathbf{h}_{\text{e}}^u & \text{if} ~\exists ~w \in \mathcal{K} \backslash \{ u \} : m_w = m_u
			\end{cases}
		\end{equation}
		and, as a consequence, equation~\eqref{SINRZF_Beams} becomes
		\begin{equation}
			\gamma^u(\mathbf{n}, \mathbf{m})  = 
			\begin{cases}
				\dfrac{\norm{\mathbf{h}_e^u}^2}{\sigma_{\mathbf{\tilde{n}}}^2}, 
				& \text{if} ~m_u \neq m_w~\forall w \in \mathcal{K} \backslash \{ u \} \\
				0 & \text{if} ~\exists ~w \in \mathcal{K} \backslash \{ u \} : m_w = m_u
			\end{cases}
		\end{equation}
		whose expected value is as \eqref{SINR_Bin_Expected}, which concludes the proof.
		\end{proof}
		\begin{remark}
			In the large-dimensional regime, the dependence of the SINR in \eqref{SINR} on the transmit beams of the other UEs
			vanishes. In particular, catastrophic co-beam interference is experienced through intersections at the BS receive beam \emph{only}. 
			We kept the dependence in \eqref{SINR_Bin_Expected} to avoid introducing additional notation. \qed
		\end{remark}
		Using Proposition~\ref{PropSINR_Bin_Expected}, the expectation in \eqref{Coo_Alg} can be approximated as follows:
		\begin{align} \label{EqSINR_Coo_K}
		\mathbb{E}_{\mathbf{n}, \mathbf{m}|\ubar{n}_u,\ubar{m}_{{1}, \dots, u}} \big[ \gamma^u(\mathbf{n}, \mathbf{m}) \big]
		\approx\sum_{\substack{(n_u, m_u) \in \mathcal{S}(\ubar{n}_u, \ubar{m}_u) \\ 
		m_u \notin \cup_{i=1}^{u-1}\mathcal{S}_{\textnormal{BS}}(\ubar{m}_{i})}} 
		\frac{g_{n_u, m_u} }{S_u \sigma_{\mathbf{\tilde{n}}}^2}.
		\end{align}
		Using \eqref{EqSINR_Coo_K} in \eqref{Coo_Alg} to choose the sub-$6$ GHz beams at the $u$-th UE allows to take 
		into account the \emph{potential} co-beam interference transferred to the lower-ranked UEs with low complexity.
		\begin{remark}
			The $K$-th (highest-ranked) UE has to consider via \eqref{EqSINR_Coo_K} the coarse-grained 
			beam decisions of all the other (lower-ranked) UEs to avoid generating potential co-beam interference. 
			Therefore, such UE might be forced to exchange high data rate for less leakage, as the best 
			non-interfering paths might have been already exploited. Therefore, it is essential to \emph{change the hierarchy} at 
			regular intervals to ensure an average acceptable rate per UE. \qed
		\end{remark}
		
		We summarize the proposed coordinated beam selection in Algorithm \ref{Coo_PseudoCode}.
		The algorithm is compatible with vectorization and parallelization, which minimize computational time.
		
		\algnewcommand\algorithmicswitch{\textbf{switch}}
		\algnewcommand\algorithmiccase{\textbf{case}}
		\algdef{SE}[SWITCH]{Switch}{EndSwitch}[1]{\algorithmicswitch\ #1\ \algorithmicdo}{\algorithmicend\ \algorithmicswitch}%
		\algdef{SE}[CASE]{Case}{EndCase}[1]{\algorithmiccase\ #1}{\algorithmicend\ \algorithmiccase}%
		\algtext*{EndSwitch}%
		\algtext*{EndCase}%
		\makeatletter \xpatchcmd{\algorithmic}{\itemsep\z@}{\itemsep=0.71ex}{}{}\makeatother
		\begin{algorithm} 
			\caption{OOB-Aided Hierarchical Coordinated Beam Selection at the generic $u$-th UE} 
			\label{Coo_PseudoCode}
			\begin{algorithmic}[1]
				\small
				\Statex INPUT: $\mathbb{E}\big[|\ubar{\mathbf{S}}^u|^2\big]$, $\ubar{m}_{1,\dots,u-1}$
				\Statex \textbf{Step 1: Exploiting OOB side-information}
					\If{$u = 1$} \Comment{The $u$-th UE is the lowest in the \emph{hierarchy}}
							\State $\mathbb{E}\big[\boldsymbol{\gamma}^u\big] = 
							\mathbb{E}\big[|\ubar{\mathbf{S}}^u|^2\big] / \sigma^2_{\mathbf{\tilde{n}}}$
							\Comment {Solve \eqref{UN_Alg} via \eqref{EqSINR_Un_K}}
					\Else \Comment{The $u$-th UE is \emph{not} the lowest in the \emph{hierarchy}}
							\For {$\ubar{n} = 1:M_{\text{UE}}$}
								\For {$\ubar{m} = 1:M_{\text{BS}}$}
									\State $N = \text{card}\big(\mathcal{S}(\ubar{n}, \ubar{m}) ~\backslash~
									\mathcal{S}_{\textnormal{BS}}(\ubar{m}_1) \cup \dots \cup
									\mathcal{S}_{\textnormal{BS}}(\ubar{m}_{u-1}) \big)$
									\State $S = \text{card}\big(\mathcal{S}(\ubar{n}, \ubar{m})\big)$
									\State $T =\mathbb{E}\big[|\ubar{\mathbf{S}}^u_{\ubar{n}, \ubar{m}}|^2 \big]
									/ \sigma^2_{\mathbf{\tilde{n}}}$
									\State $\mathbb{E}\big[\gamma^u(\ubar{n}, \ubar{m})\big] = NT/S$
									\Comment {Solve \eqref{Coo_Alg} via \eqref{EqSINR_Coo_K}}
								\EndFor
							\EndFor
					\EndIf
					\State \textbf{return} $(\ubar{n}_u^{\text{co}}, \ubar{m}_u^{\text{co}}) \leftarrow 
					\argmax_{\ubar{n}, \ubar{m}} ~\mathbb{E}\big[\boldsymbol{\gamma}^u\big]$
					\Statex \textbf{Step 2: Pilot-training the subset of mmWave beams}
					\State $(n_u^{\text{co}}, m_u^{\text{co}}) \leftarrow 
					\argmax_{n, m} \big|\mathbf{w}^u_{m} \mathbf{H}^u \mathbf{v}^u_{n} \big|^2 
					~\forall n, m \in \mathcal{S}(\ubar{n}_u^{\text{co}}, \ubar{m}_u^{\text{co}})$
			\end{algorithmic}
		\end{algorithm}
		
\section{Simulation Results} \label{sec:sims}

	We evaluate here the performance of the proposed algorithm for $K = 5$ \emph{closely-located} UEs.
	We assume $N_{\text{BS}} = 64$, $N_{\text{UE}} = 16$ for mmWave communications, and 
	$\ubar{N}_{\text{BS}} = 8$ and $\ubar{N}_{\text{UE}} = 4$ for sub-$6$ GHz ones.
	As for the carrier frequencies, we consider $28$ GHz and $3$ GHz for mmWave and sub-$6$ GHz operation, respectively.
	All the plotted data rates are the averaged -- over $10000$ Monte-Carlo iterations -- instantaneous sum-rates, obtained after ZF 
	combining at the digital stage (BS side). 
	
	\subsection{Multi-Band Channels}
	
		The performance of the proposed OOB-aided algorithms depends on the \emph{spatial congruence} between sub-$6$ GHz and mmWave channels.
		The authors in~\cite{Ali2018} proposed a simulation environment for generating sub-$6$ GHz and mmWave channels based on the model in
		\eqref{H}. The \textsc{Matlab}\textsuperscript \textregistered code used to simulate those channels is open-source and available on
		IEEEXplore~\cite{Ali2018}. We use the same model except that we consider a narrowband channel model, for which path time spread and 
		beam squint effect can be neglected~\cite{Akdeniz2014}. Note that frequency-selective filters at the BS side helps discriminating 
		(in time) among UEs which generate co-beam interference, and thus might results in giving an extra performance in average wideband 
		channels. In this paper, we consider a worst case scenario to present the substance of our idea. 
		In principle, models and algorithms could be extended to a wideband setting.
		
	\subsection{Results and Discussion}
	
		We consider a stronger (on average) LOS cluster with respect to the reflected ones, as the LOS is indeed the prominent
		propagation driver in mmWave bands~\cite{Akdeniz2014}. In particular, we adopt the following large-scale pathloss model:
		\begin{equation}
			\text{PL}(\delta) = \alpha + \beta \log_{10}(\delta) + \xi \qquad [\text{dB}]
		\end{equation}
		where $\delta$ is the path length and where the pathloss parameters $\alpha$, $\beta$ and $\xi$ are taken from 
		Table 1 in~\cite{Akdeniz2014} for both LOS and NLOS contributions. 
		The large-scale pathloss is then reflected in the cluster power $\sigma_c^2~\forall c$. 
		The average power of all the paths in a given cluster is assumed to be equal.
		Since the model in~\cite{Ali2018} is for a single-user scenario, we consider the model in~\cite{Mahmoud2002} to extend it 
		so as to generate correlated channel clusters for all the neighboring UEs in the disk. In~\cite{Mahmoud2002}, the position 
		of the clusters is made also dependent on the position of the UEs, and as a result, the possible sharing of reflectors and
		scatterers for neighboring UEs is taken into account. An example of the available sub-$6$ GHz spatial spectrum at two UEs
		is shown in Fig.~\ref{fig:ss_cm}.
		
		\begin{figure}[h]
			\centering
			\includegraphics[trim=4cm 8cm 4cm 8.5cm, scale=0.78]{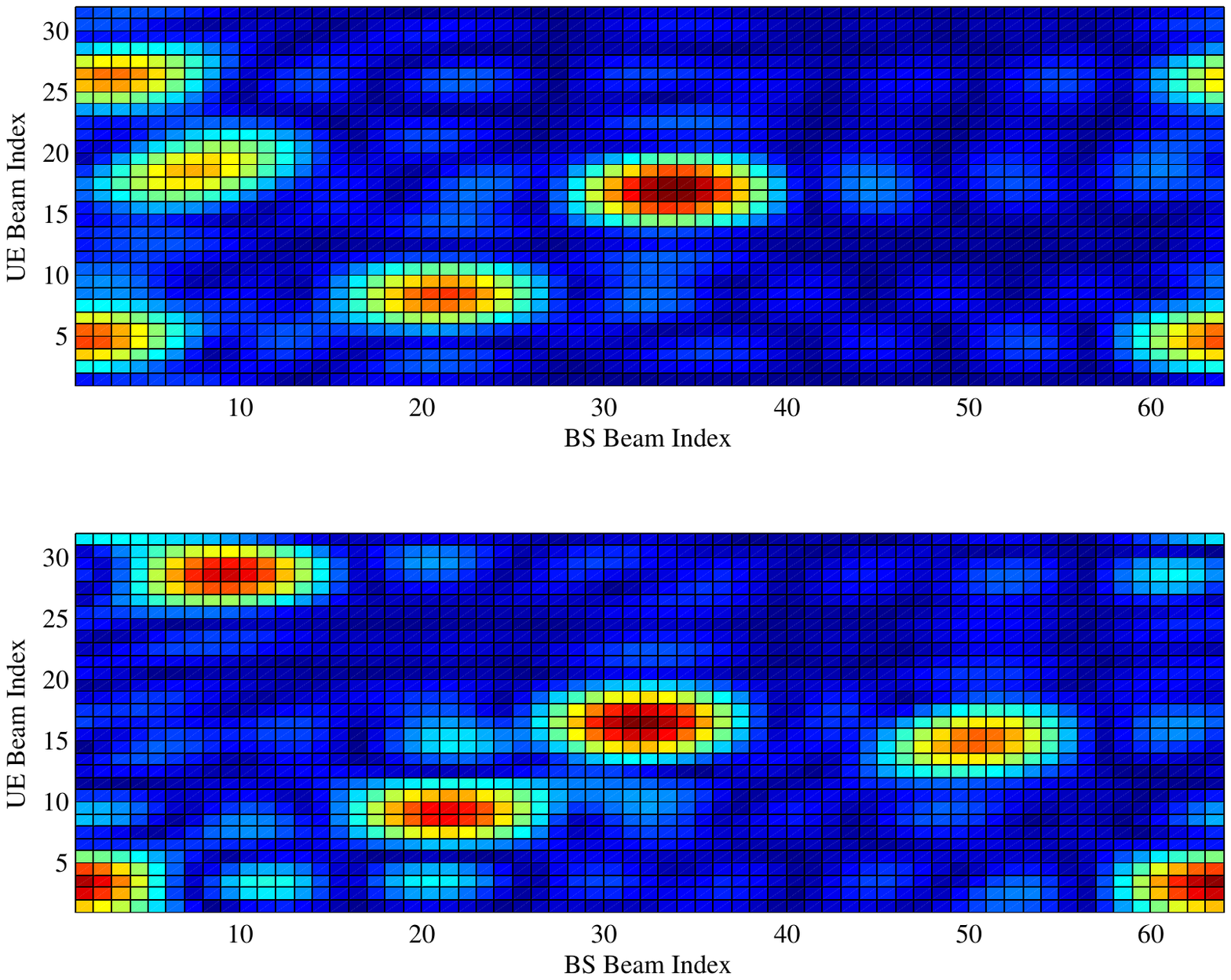}
			\caption{Example of available $\mathbb{E}[|\ubar{\mathbf{S}}^u|^2]$ at two neighboring UEs, with $r = 11$ m. 
			Some reflectors are being shared, while others are uncommon.
			The average path gains can be different.}
			\label{fig:ss_cm}
		\end{figure}
		
		In Fig.~\ref{fig:sum_rate_vs_snr}, we show the sum-rate of the proposed algorithms as a function of the SNR, where the average
		distance between the UEs is $13$ meters. For reference, we also plot the curve related to the upper bound achieved with no multi-user
		interference. The proposed OOB-aided coordinated algorithm outperforms the uncoordinated one, which neglects co-beam interference. 
		
		\begin{figure}[h]
			\centering
			\includegraphics[trim=1.57in 3.17in 1.57in 3.57in, scale=1.07]{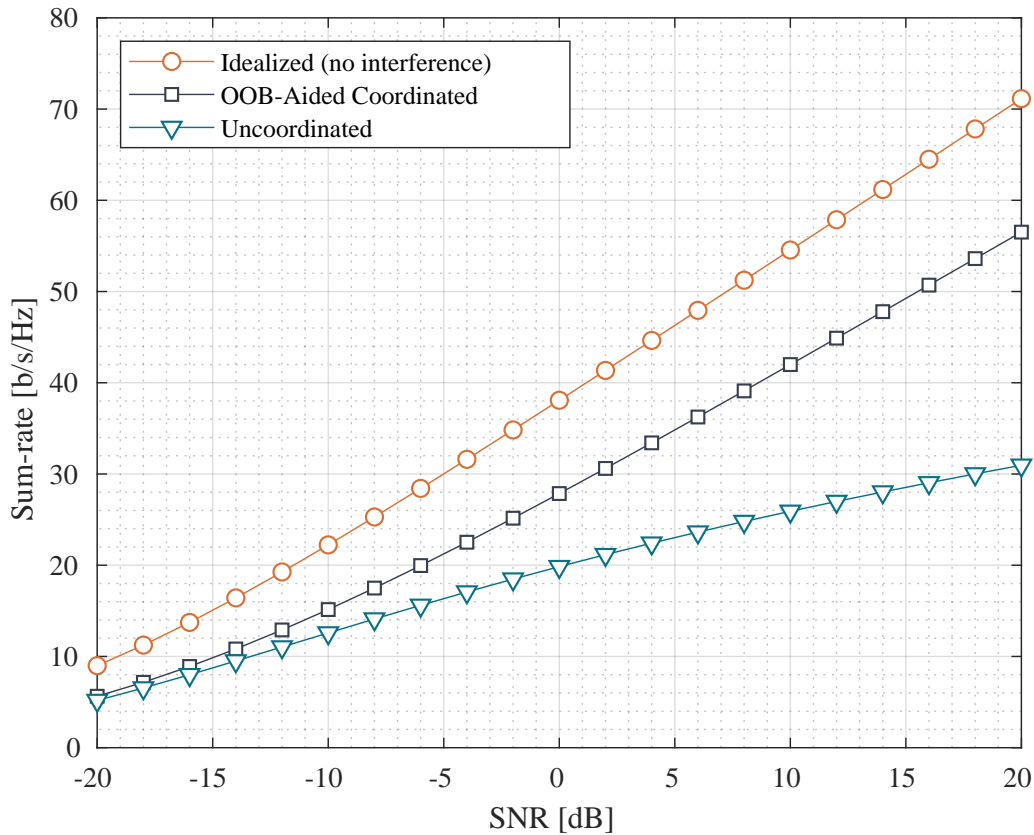}
			\caption{Sum-rate vs SNR. The average inter-UE distance is $13$ m. The OOB-aided coordinated algorithm outperforms the uncoordinated one.
			The coordination gain increases with the SNR.}
			\label{fig:sum_rate_vs_snr}
		\end{figure}
		
		In Fig.~\ref{fig:sum_rate_vs_crad}, we show the sum-rate of the proposed algorithms as a function of the average inter-UE distance, for a
		mmWave SNR of $1$ dB. 
		The coordination among the UEs allows for huge SE gains for inter-UE distances below $15$ meters. As the average inter-UE distance 
		increases -- and so, there is less chance for the co-beam interference to occur -- the performance gap between the two algorithms narrows.
	
		\begin{figure}[h]
			\centering
			\begin{overpic}[trim=1.57in 3.17in 1.57in 3.57in, scale=1.027]{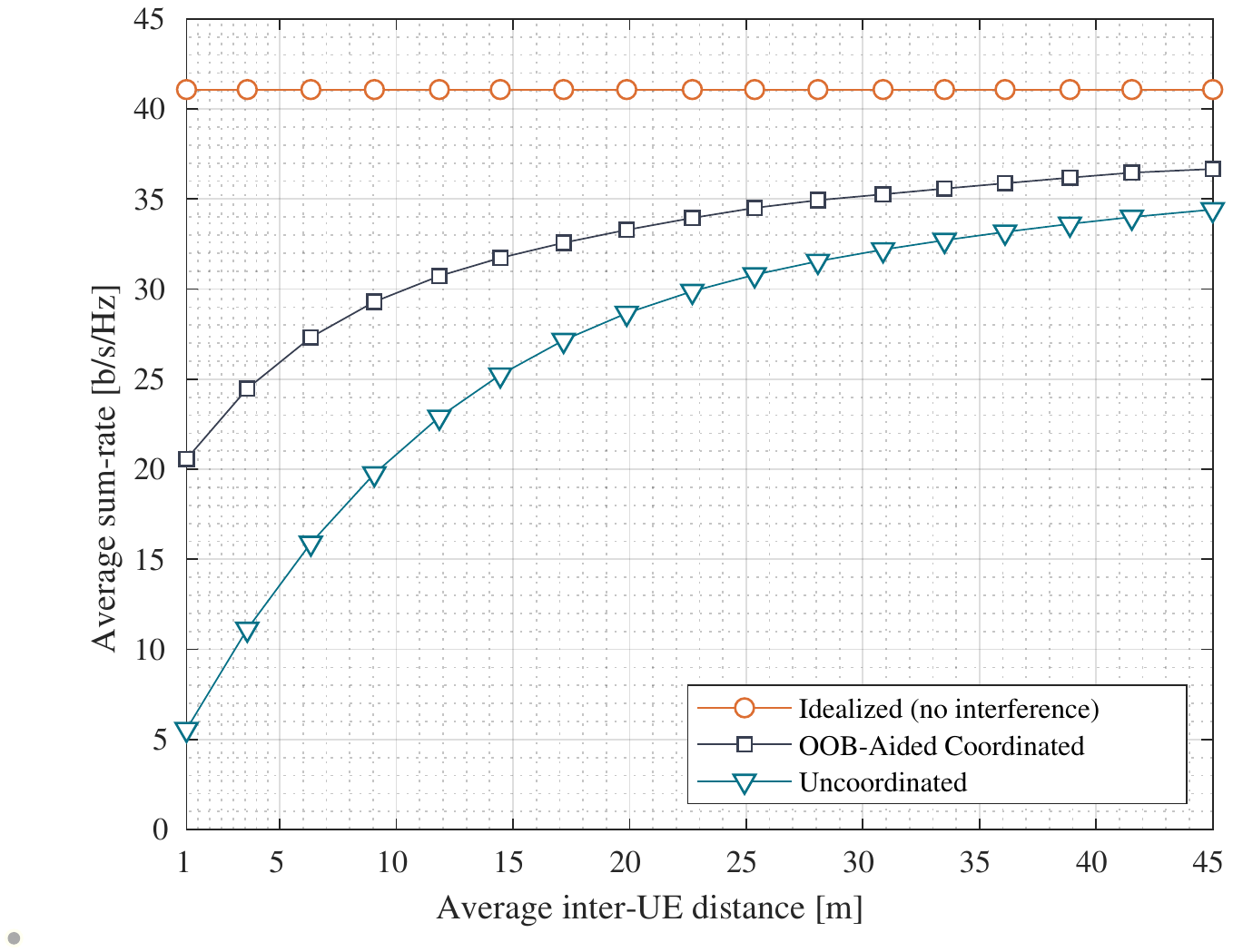}
				\put(-4.5,2){\textcolor{white}{\textbullet}}
				\put(-5,2){\textcolor{white}{\textbullet}}
				\put(-5.5,2){\textcolor{white}{\textbullet}}
			\end{overpic}
			\caption{Sum-rate vs average inter-UE distance. The SNR is fixed to $1$ dB.
			The performance gain achieved through coordination decreases with the inter-UE distance.}
			\label{fig:sum_rate_vs_crad}
		\end{figure}

\section{Conclusions}
	
	In mmWave communications, suitable strategies for interference minimization can be applied in the beam domain through
	e.g. exploiting spatial side-information. In this work, we introduced a low-overhead OOB-aided decentralized beam selection algorithm for a
	mmWave uplink multi-user scenario, leading to improved interference management. The core of the proposed algorithm resides in
	the D2D-enabled hierarchical information exchange, which allows for a low-overhead approach to the beam selection problem.
	Finding clear relationships between mmWave and lower bands radio environments is essential for OOB-aided approaches -- in particular,
	towards robust algorithms taking channels discrepancies into account -- and it is an interesting open research problem.
	
\section{Acknowledgment}

	The authors are supported by the ERC under the European Unions's Horizon 2020 research and innovation program 
	(Agreement no. 670896 PERFUME).
					
\bibliography{Bibl}
\bibliographystyle{IEEEtran}
				
\end{document}

%% file: Definitions.tex
\newtheorem{definition}{Definition}
\newtheorem{proposition}{Proposition}
\newtheorem{lemma}{Lemma}

\DeclareMathAlphabet{\mathbit}{OML}{cmr}{bx}{it}
\DeclareMathAlphabet{\mathsf}{OT1}{cmss}{m}{n}
\DeclareMathAlphabet{\mathTXf}{OT1}{cmss}{bx}{it}

\DeclareMathOperator*{\argmax}{argmax}

\newcommand{\norm}[1]{\lVert{#1}\rVert}

\theoremstyle{remark}
\newtheorem{remark}{Remark} 
\theoremstyle{example}
\theoremstyle{assumption}